\documentclass{article}
\usepackage{hyperref}
\usepackage{amsmath,amssymb}
\usepackage{amsthm}
\usepackage{graphicx}
\usepackage{physics}
\usepackage{geometry}

\graphicspath{ {im/} }
\usepackage[utf8]{inputenc}

\graphicspath{ {images/} }
\newtheorem{lemma}{Lemma}
\newtheorem{theorem}{Theorem}

\begin{document}

\begin{center}

\Large \textbf{A Generalization of Bernstein-Vazirani Algorithm to Qudit Systems}
\vspace{0.6cm}

\normalsize \textbf{R. Krishna\renewcommand{\thefootnote}{$\star$}\footnote{Email: rajath.r@xaviers.edu.in}, V. Makwana\renewcommand{\thefootnote}{$\dagger$}\footnote{Email: vishesh.makwana@xaviers.edu.in}, A. Suresh\renewcommand{\thefootnote}{$\ddagger$}\footnote{Email: ananda.suresh@xaviers.edu.in}}
\vspace{0.3cm}
\\ \textit{Department of Mathematics, \\St. Xavier's College, \\ Mumbai, India.}
\vspace{0.4cm}

\vspace{0.2cm}
\end{center}

\begin{abstract}
A quantum algorithm to solve the parity problem is better than its most efficient classical counterpart with a separation that is polynomial in the number of queries. This was shown by E. Bernstein and U. Vazirani and was one of the earliest indications that the quantum information processing can outperform the classical one by a significant margin. The problem and its solution both is usually stated for a 2-level system since we generally work with bits/qubits. However, many works have been done generalizing known quantum computing techniques to higher level systems. Following this, we look at a generalization of the Bernstein-Vazirani algorithm implemented on a general qudit system.
\end{abstract}
\vspace{0.4cm}

\markboth{R.Krishna, V.Makwana, A.Suresh}
{Generalization of Bernstein-Vazirani Algorithm} 
\vspace{0.6cm}
\section{Introduction}

Advances in quantum information and computing has bought on a shift in the way we think about information processing. Taking advantage of quantum phenomenon like superposition and entanglement, one can extend the limits of computation that the classical complexity theory predicts\cite{am}. Solution to many problems, which the classical theory predicts cannot be improved beyond a certain limit have shown to be significantly faster with the quantum approach. Shor's factoring algorithm \cite{shor} and Grover's search algorithm \cite{gr} are examples of this speed-up.\cite{zg}

The first indication of the superiority of quantum computers in being able to solve certain tasks came from Deutsch and Jozsa \cite{dj}. Their paper contained strong hints pointing towards the exponential speed-up capabilities of a quantum system when compared to the classical one. Soon after, Bernstein and Vazirani made this more concrete by stating a problem and demonstrating its quantum solution to be having query complexity significantly less than what was possible by any classical approach\cite{bv}. The solution to their problem which had a query complexity of $\Omega(n)$ by the best possible classical algorithm was shown to have $\Omega(1)$ in the quantum realm(polynomial speed-up). Even more impressive was a different version of the problem, called the recursive Bernstein-Vazirani problem which had an exponential complexity in the classical world and a logarithmic one in the quantum(super polynomial speed-up). 

In addition to figuring out quantum algorithms to do various tasks, over the years, people have also worked out how to work with a general `qudit' system instead of the usual 2 level qubit \cite{gg}. This means more information can be stored using fixed number of qudits. Physically this can be implemented, for example, with three energy levels of a hydrogen atom for a qutrit system. The question of how to generalize known quantum algorithms for a qudit system is hence, of theoretical and practical importance\cite{gdj}. Here the Bernstein-Vazirani problem is considered and its generalization to a higher dimensional system is derived and discussed.
\section{The Bernstein-Vazirani problem}

In the Bernstein-Vazirani problem we are given an oracle of the form, $$f_s:\{0,1\}^n\rightarrow\{0,1\};\ \  s\in \{0,1\}^n$$ This takes as input an n-bit string and outputs a single bit. The function is defined to be, $$f_s(x)=s\cdot x \hspace{0.5cm} \mathrm{where} \hspace{0.2cm} x \in \{0,1\}^n$$ Here $s$ is an unknown string and $$s\cdot  x=\left(\sum_{i=1}^{n}s_ix_i\right)\hspace{3pt}\mathrm{mod}\hspace{5pt}2$$ where $s_i$ and $x_i$ are $\mathrm{i^{th}}$ bits of $s$ and $x$ respectively. The Bernstein-Vazirani problem is to find the unknown string $s$ by querying the oracle.

\subsection{A classical algorithm}

To find the hidden string, all we are allowed to do is to query the function with different inputs. To find the first bit of the string, we can use the state `100..0' as our query. If the first bit is 1, the function will output 1 and if 0, it will output 0. Similarly, to find the second bit we use the state `010...0' as our query. Proceeding this way, we can find all bits of the unknown string s. This means if the string has a length n, using this algorithm, we would require n queries to find it. That is, this algorithm has a classical query complexity of $\Omega(n)$. In addition to this, we also notice that since each bit in the string is completely independent of all the other bits, there is no other algorithm that provide a query complexity better than this.    

\section{Quantum Algorithm}

Even though classical complexity theory limits the query complexity by $\Omega(n)$ the Bernstein-Vazirani's solution to the same problem is shown to have quantum query complexity of $\Omega(1)$. This is unexpected and bizarre power of quantum information processing. Described below is this algorithm. 
\subsection{The Bernstein-Vazirani algorithm for a 2-level system}

Perhaps the most crucial step used in any quantum algorithm is Fourier sampling. It involves applying the Hadamard gate to each qubit of an n-qubit system. Fourier sampling on a general computational basis state of an n-qubit system is given by the following transformation,
\begin{equation}
H^{\otimes n}\ket{u}=\frac{1}{\sqrt{2^n}}\sum_{x\in \{0,1\}^n}(-1)^{u\cdot x} \ket{x}
\end{equation}
where $H^{\otimes n}$ indicates the operator obtained by tensoring $H$ n times.

Now, we use the interesting property of Quantum Gates that they are reversible and the Hadamard turns out to be its own inverse. Thus,
$$H^{\otimes n} \left(\frac{1}{\sqrt{2^n}}\sum_{x\in \{0,1\}^n}(-1)^{u\cdot x} \ket{x}\right)=\ket{u} $$

\hspace{-0.53cm}In the Bernstein-Vazirani problem as described in above section, the aim is to find the unknown string $s$. It is easy to see that the Fourier Sampling on the superposition $$\frac{1}{\sqrt{2^n}}\sum_{x\in \{0,1\}^n}(-1)^{s\cdot x} \ket{x}$$ will give $s$. Thus, the problem drops down to setting up the above superposition using the quantum oracle and Fourier sampling it to obtain $s$. 
\subsubsection{Setting up the superposition}
Consider applying Fourier transformation on an n-bit string $\ket{00..0}$. Its output will be an equal superposition of all possible n-bit strings, all having positive phase \cite{ln}. 

$$H^{\otimes n} \ket{00..0} = \frac{1}{\sqrt{2^n}} \sum_{x \in \{0,1\}} \ket{x} $$
\begin{center}
\hspace{-0.6cm}\includegraphics[width=0.5\textwidth]{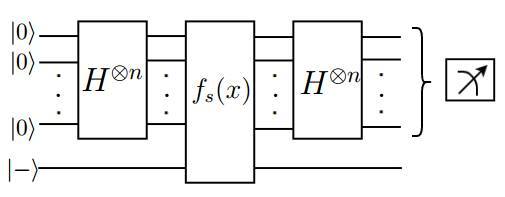}
\vspace{0.35cm}

\hspace{-0.18cm}\small Fig. 1. Quantum Circuit for Berstein Vazi-

\hspace{1.4cm}rani algorithm for a two-level system.
\vspace{0.35cm}
\end{center}
\normalsize This is the input we feed into the oracle $f_s(x)$ as our query. We then take the output from $f_s(x)$ and use it as the control bit in a CNOT operation whose target bit is $\ket{-} = \frac{1}{\sqrt{2}}(\ket{0} - \ket{1})$. 

If the control bit is one, that is, if the output of $f_s(x)$ is one, the CNOT will act on $\ket{-}$ and change it to $-\ket{-}$ and if it is zero, then the target bit will stay as it is. Now, in a tensor product of two quantum states we are free to associate the sign with whichever state we choose to. $$\ket{u} \otimes (-\ket{v})=-(\ket{u} \otimes \ket{v}) = (-\ket{u}) \otimes \ket{v}$$ This means we can associate the negative sign of $\ket{-}$ with $\ket{x}$ in the superposition. This amounts to flipping the phase of the terms of the superposition whenever $u \cdot x = 1$.
In other words, we create the superposition $$\frac{1}{\sqrt{2^n}} \sum_{x \in \{0,1\}^n} (-1)^{s\cdot x} \ket{x}$$ 

\hspace{-0.53cm}In summary these are the transformations to setup the required superposition

$$\ket{00...0} \xrightarrow[]{{\cal H}} \frac{1}{\sqrt{2^n}}\sum_x \ket{x} \xrightarrow[]{\mathrm{use} f_s(x)} \frac{1}{\sqrt{2^n}}\sum_x (-1)^{s \cdot x} \ket{x}$$

\hspace{-0.53cm}As explained earlier, taking the Fourier transform of the last state would give the state $\ket{s}$, which we can then measure.                 

\subsection{Generalized quantum states and gates}

The generalized form of the quantum circuit will be one in which all the quantum states and gates will be replaced by their corresponding higher dimensional versions. The d-dimensional quantum state, instead of being a superposition of the kets $\ket{0}$ and $\ket{1}$, will be one that of the kets $\ket{0}, \ket{1}...\ket{d-1}$. Its phase is going to be of the form $\omega^i$ where $\omega$ is the $dth$ root of unity instead of the usual $\pm1$. During measurement, the generalized state behaves the same way as the normal state does. It collapses into any of its basis states with a probability given by the   absolute square of its coefficient. The quantum states being able to take up more than two levels means more information condensed into a single state. 

Next logical step in our process is to design higher dimensional quantum gates capable of manipulating these sates. Out of many generalizations possible, here we look at what is perhaps the simplest and most straightforward\cite{dg}. The motivation behind defining gates this way is that they allow for the creation of a higher dimensional analogue of the quantum circuit without much modification to the original one. 

Under this, the Hadamard gate $H$ becomes the d-dimensional discrete Fourier transform matrix defined by, 
\begin{equation}
\ket{j} \rightarrow \sum_{s=0}^{d} \omega^{js}\ket{s}
\end{equation}

\hspace{-0.53cm}CNOT becomes the SUM gate. 
\begin{equation}
\ket{i}\ket{j} \rightarrow \ket{i}\ket{(i+j)\hspace{0.1cm} \mathrm{mod} \hspace{0.1cm} d}
\end{equation}

\subsection{Quantum algorithm for a d-level system}

In this section, we provide an algorithm to solve the Bernstein-Vazirani problem for a d-dimensional system. Here the problem changes to finding an unknown string $s \in \{0,1...(d-1)\}^n$ by querying a function, $$f_s(x) = s \cdot x \hspace{0.1cm} \mathrm{mod} \hspace{0.1cm} d$$ given as a quantum oracle.  In order to achieve this we set up the superposition $$\ket{\psi_s} = \frac{1}{\sqrt{d^n}} \sum_x \omega^{x \cdot s} \ket{x}$$ 
${x \in \{0,1...(d-1)\}^n}$ 

\hspace{-0.53cm}using the given oracle. We then apply the tensor product of the d-dimensional Fourier transform into this state to obtain back $s$. 

Below we prove that Fourier sampling on the superposition will give back $s$ and we explain how to set up the superposition.

\begin{lemma}
\label{lemma1}
Primitive $d^{th}$ root of unity satisfy $\sum\limits_{\alpha =0}^{d-1} \omega^{\alpha k} = d\cdot \delta_{k,0}$ where $k \in \mathbb{N}\hspace{0.05cm}\cup\{0\}$
\end{lemma}

\begin{proof}
\hspace{0.6cm}Consider the summation for a fixed $k$,$$\sum\limits_{\alpha =1}^{d-1} \omega^{\alpha k} = 1+\omega^{k}+\omega^{2k}+...+\omega^{(d-1)k}$$

\hspace{-0.53cm}For $k=0$, it is easy to see that this is equal to $d$.

\hspace{-0.53cm}When $k \neq 0$,
$$\sum\limits_{\alpha=1}^{d-1} \omega^{\alpha k} =\frac{1-(\omega^{k})^d}{1-\omega^{k}}=\frac{1-(\omega^d)^{k}}{1-\omega^{k}}=\frac{1-1}{1-\omega^{k}}=0$$

\hspace{-0.53cm}Or in other words,
\begin{equation}
\sum\limits_{\alpha=1}^{d-1} \omega^{\alpha k} = d \cdot \delta_{k,0} 
\end{equation}
\end{proof}

\subsubsection{Fourier sampling the superposition}
\begin{theorem}
\label{theorem}
The following quantum states defined by,$$\ket{\psi_s} = \frac{1}{\sqrt{d^n}} \sum_x \omega^{s \cdot x} \ket{x}$$ 
\textit{where}  $x\in \{0,1...(d-1)\}^n$\textit{are orthogonal to each other for different values of $s$.}
\end{theorem}
\vspace{0.5cm}
\begin{proof}
\hspace{0.6cm}Consider,$$\ket{\psi_s} = \frac{1}{\sqrt{d^n}} \sum_x \omega^{s \cdot x} \ket{x}$$ 
$$\ket{\psi_t} = \frac{1}{\sqrt{d^n}} \sum_y \omega^{t \cdot y} \ket{y}$$ 
\hspace{4.73cm}$\braket{\psi_s}{\psi_t} = \frac{1}{d^n} \sum_x (\omega^*)^{x \cdot s} \bra{x} \frac{1}{d^n} \sum_y \omega^{y \cdot t} \ket{y}$
$$=\frac{1}{d^n} \sum_{x,y} (\omega^*)^{x \cdot s} (\omega)^{y \cdot t}  \braket{x}{y}$$
$$=\hspace{0.25cm} \frac{1}{d^n} \sum_{x,y} (\omega^*)^{x \cdot s} (\omega)^{y \cdot t} \delta_{x,y}$$
\begin{equation}
\hspace{-0.82cm}
=\frac{1}{d^n} \sum_x (\omega^*)^{x \cdot s} (\omega)^{x \cdot t}   
\end{equation}

\hspace{-0.66cm} As $\omega \omega^* = 1$, we have,

\hspace{-0.53cm}Case 1: $x \cdot s > x \cdot t$

$$\hspace{-1cm}\braket{\psi_s}{\psi_t}=\frac{1}{d^n} \sum_x (\omega^*)^{x \cdot s - x \cdot t} $$ 

\begin{equation}
\hspace{0.15cm}
=\frac{1}{d^n} \sum_x (\omega^*)^{x \cdot (s-t)}
\end{equation}

\hspace{-0.53cm}Case 2: $x \cdot s < x \cdot t$ 

$$\braket{\psi_s}{\psi_t}=\frac{1}{d^n} \sum_x (\omega)^{x \cdot t-x \cdot s} $$ 

\begin{equation}
\hspace{1.05cm}
=\frac{1}{d^n} \sum_x (\omega)^{x \cdot (t-s)}
\end{equation}

\hspace{-0.53cm}Now consider the sum  for a fixed $k$, $$\sum_x \omega^{x \cdot k} = \sum_x (\omega)^{x_1k_1} (\omega)^{x_2k_2} \dotso (\omega)^{x_nk_n}$$
\begin{equation}
\hspace{1.6cm}
= \sum_{x_1} \omega^{x_1k_1}\sum_{x_2} \omega^{x_2k_2} \dotso \sum_{x_n} \omega^{x_nk_n}
\end{equation}

\hspace{-0.53cm}Here $x_i,k_i \in \{0,1...(d-1)\}$

\vspace{0.3cm}
\hspace{-0.53cm}Using eq. (4) in eq. (8) we get,
$$\sum_x \omega^{x \cdot k} = (d \cdot \delta_{k_1,0}) (d \cdot \delta_{k_2,0})...(d \cdot \delta_{k_d,0})=d^n \delta_{k,0}$$ 

\hspace{-0.53cm}Plugging this into (5) and (6) with $k=s-t$ and $k=t-s$ respectively we have, $$\braket{\psi_s}{\psi_t}= \delta_{s-t,0}\quad \mathrm{when} \hspace{0.2cm} x \cdot s>x\cdot t$$
$$\braket{\psi_s}{\psi_t}= \delta_{t-s,0}\quad \mathrm{when} \hspace{0.2cm} x\cdot t>x\cdot s$$ 

\hspace{-0.53cm}which implies, 
\begin{equation}
\braket{\psi_s}{\psi_t} = \delta_{s,t}
\end{equation}

\hspace{-0.53cm}Or in other words, the states are orthogonal.
\end{proof}

\vspace{0.8cm}
\hspace{-0.53cm}Now,the tensor product of the d-dimensional discrete Fourier transform matrix is given by, 
\begin{equation}
{\cal F}^{\otimes n} = \frac{1}{\sqrt{d^n}} \sum_{x,y} \omega^{x \cdot y} \ketbra{y}{x}
\end{equation}

\hspace{-0.53cm}This can be expressed as $${\cal F}^{\otimes n} = \sum_y \ketbra{y}{\psi_y}$$ Now we apply ${\cal F}^{\otimes n}$ to $\ket{\psi_s}$,

$${\cal F}^{\otimes n} \ket{\psi_s} = \sum_y \ket{y}\braket{\psi_y}{\psi_s}$$which is the same as
\begin{equation}
{\cal F}^{\otimes n} \ket{\psi_s} = \ket{s}
\end{equation}
Hence we show that the Fourier sampling on the state $\ket{\psi_s}$ gives $\ket{s}$.

\hspace{-0.53cm}So if we manage to set up the superposition $$\frac{1}{\sqrt{d^n}} \sum_x \omega^{s \cdot x} \ket{x}$$
using the given oracle, we can then Fourier sample it to get the required hidden string $s$.

\subsubsection{Setting up the superposition}

The first step in setting up the superposition is Fourier sampling the state $\ket{00..0}$, which gives $$\ket{\Psi} = \frac{1}{\sqrt{d^n}}\sum_x \ket{x} \hspace{1cm}; \hspace{0.2cm} x \in \{0,1...(d-1)\}^n$$
\begin{center}
\hspace{-0.6cm}\includegraphics[width=0.5\textwidth]{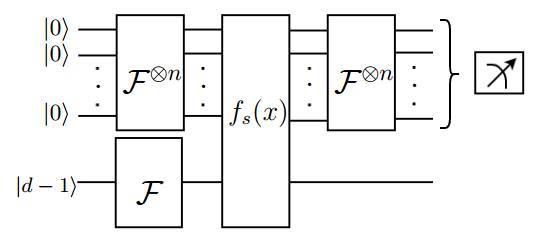}
\vspace{0.35cm}

\hspace{-0.18cm}\small Fig. 2. Quantum Circuit for Berstein Vazi-

\hspace{1.4cm}rani algorithm for a d-level system.
\vspace{0.35cm}
\end{center}

\normalsize \hspace{-0.53cm}We now take the oracle $f_s(x)$ and feed this superposition as its query. The output from $f_s(x)$ is used as the control bit of a SUM gate whose target bit is the superposition
\begin{equation}
\ket{\phi}=\omega^d \ket{0} + \omega^{d-1} \ket{1}+...+\omega \ket{d-1}
\end{equation}

\vspace{0.3cm}
\hspace{-0.53cm}Now consider applying the SUM gate to the state $\ket{\phi}$. Each term in $\ket{\phi}$ is of the form $\omega^{d-j}\ket{j}$.

\begin{equation}
\omega^{d-j}\ket{j} \xrightarrow[]{\mathrm{SUM}_i} \omega^{d-j}\ket{(j+i) \hspace{0.1cm} \mathrm{mod} \hspace{0.1cm} d}
\end{equation}

\hspace{-0.53cm}Setting $i+j = k \Rightarrow d-j=d+i-k$

\hspace{-0.53cm}Hence (9) becomes, $$\omega^{d-j} \ket{j}\xrightarrow[]{\mathrm{SUM}_i}
\omega^i \omega^{d-k} \ket{k\hspace{0.1cm} \mathrm{mod} \hspace{0.1cm} d}$$

\hspace{-0.53cm}Now, when $k<d$ we have $\ket{k\hspace{0.1cm} \mathrm{mod} \hspace{0.1cm} d}=\ket{k}$ and thus, we get 
\begin{equation}
\omega^{d-j} \ket{j}\xrightarrow[]{\mathrm{SUM}_i} \omega^i \omega^{d-j}\ket{j}
\end{equation}
as $k$ is just a dummy variable.

\hspace{-0.53cm}Also, as $i$ and $j$ are bounded above by $d-1$, $k$ is strictly less than $2d$. Hence, when $d\leq k< 2d$ we have $\ket{k\hspace{0.1cm} \mathrm{mod} \hspace{0.1cm} d}=\ket{k-d}$. 

\hspace{-0.53cm}Now, if we take $k-d=m$ then, $$\omega^{d-k} \ket{k}=\omega^{-m}\ket{m}=\omega^{d-m}\ket{m}$$ and we get back equation (10).

\hspace{-0.53cm}Hence applying the SUM gate on $\ket{\phi}$ will give,
$$\mathrm{SUM} \hspace{0.1cm} \ket{\phi} = \omega^i\ket{\phi}=\omega^{s \cdot x} \ket{\phi}$$

\hspace{-0.53cm}Just as we did in the case for a 2-level system, we can associate this phase with any part of the tensor product.

$$\ket{\Psi} \otimes (\omega^{s \cdot x} \ket{\phi})=\omega^{s \cdot x}(\ket{\Psi} \otimes \ket{\phi})=(\omega^{s \cdot x}\ket{\Psi}) \otimes \ket{\phi}$$  
Hence $\ket{\Psi}$ will become,
\begin{equation}
\ket{\Psi} = \frac{1}{\sqrt{d^n}}\sum_x \omega^{s \cdot x}\ket{x}
\end{equation}

\hspace{-0.53cm}As explained in the previous section, Fourier sampling this state will give $\ket{s}$. 

\section{Conclusion}

\vspace{-0.10cm}
A general Bernstein-Vazirani parity problem was defined and the corresponding quantum algorithm was generalized. This was done by replacing the quantum gates for qubits with the generalized quantum gates. It was seen that the superpositions $\ket{\psi_s}$ and $\ket{\psi_t}$ are orthogonal and this property was used to retrieve the unknown string s by Fourier Sampling them. The quantum query complexity of the original 2-dimensional problem is $\Omega(1)$. Since, even in the generalized version as only one input is to be given, the query complexity remains the same as in the 2-dimensional case.        

Algorithms such as this clearly depict the power of quantum computing. Even though quantum computing has been centered around two-level systems and most of the algorithms have been for qubits, higher dimensional systems might play a major role in the same. Studying the generalizations of already existing quantum algorithms and developing new ones for the higher level systems will help us realize the potential of qudit systems. Future work in this direction could be for example, a generalization of the recursive Bernstein-Vazirani algorithm for qudit systems.


\begin{thebibliography}{11}
\normalsize
{\bibitem{am}
Montanaro Ashley, ``Quantum algorithms: an overview'',npj Quantum Information 2, Article number: 15023 (2016)
DOI:10.1038/npjqi.2015.23

\bibitem{shor}
Peter W. Shor, ``Polynomial-Time Algorithms for Prime Factorization and Discrete Logarithms on a Quantum Computer'', SIAM J.Sci.Statist.Comput. 26 (1997) 1484.

\bibitem{gr}
Grover L.K.: ``A fast quantum mechanical algorithm for database search'', Proceedings, 28th Annual ACM Symposium on the Theory of Computing, (May 1996)

\bibitem{zg}
Z. Gedik, I. A. Silva, B. Cakmak, G. Karpat, E. L. G. Vidoto, D. O. Soares-Pinto, E. R. deAzevedo, F. F. Fanchini, ``Computational speed-up with a single qudit'', Scientific Reports 5, Article number: 14671 (2015), DOI: 10.1038/srep14671

\bibitem{dj}
David Deutsch and Richard Jozsa (1992). ``Rapid solutions of problems by quantum computation''. Proceedings of the Royal Society of London A. 439: 553. Bibcode:1992RSPSA.439..553D. DOI:10.1098/rspa.1992.0167.

\bibitem{bv}
Bernstein Ethan, Vazirani Umesh (1993).``Quantum complexity theory", Proceedings of the Twenty-Fifth Annual ACM Symposium on Theory of Computing (STOC '93), pp. 11–20, DOI:10.1145/167088.167097.


\bibitem{gg}
Erik Hostens, Jeroen Dehaene, Bart De Moor. Stabilizer states and Clifford operations for systems of arbitrary dimensions, and modular arithmetic, Phys. Rev. A 71, 042315 (2005), DOI: 10.1103/PhysRevA.71.042315


\bibitem{gdj}
Yale Fan, ``A Generalization of the Deutsch-Jozsa Algorithm to Multi-Valued Quantum Logic'', Multiple-Valued Logic 2007. ISMVL 2007. 37th International Symposium on, ISSN 0195-623X, DOI: 10.1109/ISMVL.2007.3

\bibitem{ln}
Bacon, Dave.``Lecture notes-Introduction and Basics of Quantum Theory"
http://courses.cs.washington.edu/courses/cse599d/06wi/lecturenotes7.pdf

\bibitem{dg}
Gottesman, Daniel.``Fault-tolerant quantum computation with higher-dimensional systems." Quantum Computing and Quantum Communications. Springer Berlin Heidelberg, 1999. 302-313.}



\end{thebibliography}
\end{document}